\let\oldvec\vec
\documentclass[runningheads]{llncs}
\let\vec\oldvec
\usepackage{amsmath}
\usepackage{graphicx}
\usepackage{enumitem}
\usepackage{mathrsfs}
\usepackage{multirow}
\usepackage{amsfonts}
\usepackage[colorlinks]{hyperref}
\usepackage{setspace}
\usepackage{textcomp}
\usepackage{enumerate}
\usepackage[ruled,vlined,linesnumbered]{algorithm2e} 
\usepackage[title]{appendix}
\usepackage{marvosym}
\usepackage{hyperref}

\begin{document}
\title{Testing Higher-order Clusterability on graphs \thanks{This work was supported by the National Natural Science Foundation of China under grants 61832003, Shenzhen Science and Technology Program (JCYJ202208181002205012) and Shenzhen Key Laboratory of Intelligent Bioinformatics (ZDSYS20220422103800001).}}
%
%
\author{Yifei Li\inst{1,2}\orcidID{0009-0005-6912-0582}, 
Donghua Yang\inst{1}\orcidID{0000-0002-6102-1804}\and
Jianzhong Li\inst{2}\orcidID{0000-0002-4119-0571}}
\authorrunning{Y. Li, D. Yang et al.}
%
\institute{
	\email{yf.li@stu.hit.edu.cn}\\
    \email{\{yang.dh,lijzh\}@hit.edu.cn}
    \\
	\inst{1}Harbin Institute of Technology, Harbin, Heilongjiang, China
    \\
    \inst{2}Shenzhen Institute of Advanced Technology, Chinese Academy of Sciences, Shenzhen, China
}
\maketitle 
\begin{abstract}
Analysis of higher-order organizations, usually small connected subgraphs called motifs, is a fundamental task on complex networks. This paper studies a new problem of testing higher-order clusterability: given query access to an undirected graph, can we judge whether this graph can be partitioned into a few clusters of highly-connected motifs? This problem is an extension of the former work proposed by Czumaj et al. (STOC' 15), who recognized cluster structure on graphs using the framework of property testing. In this paper, a good graph cluster on high dimensions is first defined for higher-order clustering. Then, query lower bound is given for testing whether this kind of good cluster exists. Finally, an optimal sublinear-time algorithm is developed for testing clusterability based on triangles.

\end{abstract}
\section{Introduction}\label{sec:Intro}
\subsection{Motivation}\label{sec:Moti}
In many real-world systems, interactions and relations between entities are not pairwise, but occur in higher-order organizations that are usually small connected patterns denoted as motifs, including triangles, wedges, cliques, etc. Some researches focus on higher-order clustering~\cite{benson2016}, which captures connected motifs into cohesive groups while motifs between different groups have few connections. Authors in \cite{benson2016} gave an example of clustering based on a particular triangle motif, which correctly represents three well-known aquatic layers in Florida Bay foodweb. Higher-order clustering has been widely applied in social network analysis~\cite{li2017}, gene regulation~\cite{gama2016} and neural networks~\cite{duval2022}. However, graphs such as Actors~\cite{actor1990} and Coauthoring~\cite{newman2002}, which are nearly bipartite, are not suitable for higher-order clustering based on triangles or cliques. Therefore, it is important to judge whether the given graph is suitable for clustering based on the specified motif. A helpful method is to use property testing~\cite{rubinfeld1996}, which is a framework that decides whether an object has a specific property or is "far" from objects having this property. However, none of the former property testing on graphs considered higher-order motifs.

In this paper, we develop a new framework of testing whether a given graph is higher-order clusterable, which is compatible with the low-order testing problem given by Czumaj et al.~\cite{czumaj2015}. First, what is a good high-dimension cluster is defined on undirected graphs. Requirements of the high-dimension cluster wouldn't violate the topological structure in lower-dimension. The problem of testing higher-order clusterability is then proposed. It asks whether there exists a good high-dimension cluster or is far from having that kind of cluster. Finally, A sublinear-time algorithm for testing triangle-based clusterability is developed, which reaches the lower bound, $\Omega(\sqrt{n})$, and is nearly optimal.

\subsection{Related Work}
This section reviews some previous researches and analyze their shortages or differences compared to the work in this paper.\vspace{5pt}\\
\textbf{Previous work on higher-order graph clustering.} 
Earlier researches on higher-order clustering is related to hypergraph partitioning~\cite{karypis1999}. Benson et al.~\cite{benson2016} proposed a generalized framework with motif conductance that could cluster higher-order connectivity patterns. They implemented an algorithm  without suffering hypergraph fragmentation and its time complexity is bounded by the number of motifs. Tsourakakis et al.~\cite{tsourakakis2017} shared the same contribution in parallel that a weighted graph can be used to replace the hypergraph in motif-based clustering. Li et al.~\cite{li2019} proposed an edge enhancement approach to solve the issue in isolated nodes.

However, there solutions have some drawbacks. First, whether there exists triangle expanders that are not edge expanders is still unknown. It means that triangle-based clustering in~\cite{tsourakakis2017} could still violate lower-order cluster structure. In addition, time complexity of higher-order clustering are bounded by the number of motifs, which is $\Omega(n^{3/2})$ if motif is triangle and is far from sublinear. Furthermore, they do not consider whether the given graph is suitable for higher-order clustering. As a result, an inappropriate clustering would suffer severe computation cost on huge graphs.\vspace{5pt}\\
\textbf{Previous work on graph property testing.}
Framework on testing graph properties was first proposed by Goldreich and Ron~\cite{goldreich1997}, who present an alternative model that each query on bounded-degree graph returns a vertex with one indexed neighbor. They showed that testing whether a graph is an expander requires $\Omega(\sqrt{n})$ queries under this model. Latter work~\cite{czumaj2010}~\cite{kale2011} provided optimal algorithms that reach this lower bound. Czumaj et al.~\cite{czumaj2015} defined $(k,\phi)$-clusterable graphs that can be partitioned into $k$ clusters with requirements on both internal and external conductance. They maintained a logarithmic gap between two conductance so that testing clusterability is equivalent to testing expansion when $k=1$. Chiplunkar et al.~\cite{chiplunkar2018} eliminated the logarithmic gap at the cost of rising the lower bound to $\Omega(n^{1/2+O(\epsilon)})$. Gluch et al.~\cite{gluch2021} designed a clustering oracle that allows fast query access and proposed an optimal algorithm. 

All the above testers consider the property of low-order expansion and fail to unravel higher-order organizations such as dense cluster of triangles. Furthermore, these testers adopt simple or lazy random walk that starts from vertices, which cannot catch information of triangles or k-cliques. Therefore, they cannot be easily extended to learning clusterability of higher order motifs. 

\subsection{Contributions}
Specifically, contributions are summarized as follows:
\begin{itemize}[leftmargin=12pt, itemsep=5pt, topsep=5pt, partopsep=5pt]
\item[1.] Problem of testing higher-order clusterability based on a new definition of high-dimension cluster. 
\item[2.] Proof that the redefined problem is compatible with the original one defined by Czumaj et al.~\cite{czumaj2015} 
\item[3.] An $\Omega(\sqrt{n})$ query lower bound of testing higher-order clusterability.
\item[4.] A sublinear-time algorithm for testing triangle-based clusterability, which reaches the lower bound with neighbor query oracle.
\end{itemize}
\subsection{Organization of the paper} Section \ref{sec:PaP} provides preliminary and statement of testing higher-order clusterability on bounded degree graphs. Section \ref{sec:comp} establishes relationship between higher-order clusterability and counterpart, and then gives a query lower bound. Section \ref{sec:Algo} proposes algorithms for testing triangle clusterability and analysis of correctness and running time. Section~\ref{sec:Fut} gives a summary on the whole paper and presents the future work.

\section{Preliminary and Problem Statement}\label{sec:PaP}
\subsection{Testing Clusterability on bounded-degree graphs}
Here is a brief review on the problem of testing graph cluster structure. Let $G=(V, E)$ be an undirected graph. $deg(v)$ denotes the degree of vertex $v$. For two non-empty vertex sets $S$ and $C$, $S\subset C\subseteq V$, let $Vol(S)=\Sigma_{v\in S}deg(v)$ denote the volume of set $S$. The \textit{external conductance}~\cite{kannan2004} of $S$ on $C$ is defined as 
$$\Phi_C(S)=\frac{|E(S,C\backslash S)|}{\min\{Vol(S),Vol(C\backslash S)\}},$$ 
where $E(S,C\backslash S)$ is the set of edges with two endpoints contained in set $S$ and $C\backslash S$ respectively. In addition, $G[C]$ denotes the \textit{induced graph} whose vertex set is $C$ and whose edge set consists of all edges with both endpoints included in $C$. Then the \textit{internal conductance} is defined as 
$$\Phi(G[C])=\min\limits_{\substack{\emptyset\neq S\subset C\\ |S|\leq|C|/2}}\frac{|E(S,C\backslash S)|}{Vol(S)}.$$
Since $\Phi_C(S)=\Phi_C(C\backslash S)$, usually only vertex sets $S$ with $|S|\leq |V|/2$ are considered for convenience. The definition of characterizing the cluster structure of undirected graph is shown as follows,
\begin{definition} (\textbf{($k,\phi_{in},\phi_{out}$)-cluster}~\cite{gharan2014}). 
    Given an undirected graph $G(V,E)$ with parameters $k,\phi_{in},\phi_{out}$, find an $h$-partition $\mathbb{P}$ of $V$, $\mathbb{P}=(P_1,P_2$, $\dots,P_h)$ with $1\leq h\leq k$, and for each $i$, $1\leq i \leq h$, $\Phi(G[P_i])\geq\phi_{in}$ and $\Phi_G(P_i)\leq\phi_{out}$.
\end{definition}
In the property testing framework, $G$ is given as a \textit{neighbor query oracle}. When given an index pair $(v, i)$, the oracle returns the predetermined $i$th neighbor of vertex v if $i$ doesn't exceed the degree of $v$, otherwise it would return \textit{NULL}.
\begin{definition} (\textbf{Testing ($k,\phi_{in},\phi_{out},\epsilon$)-clusterability}).~\cite{czumaj2015}
    Given a neighbor oracle access to graph $G(V,E)$ with maximum degree at most $d_{max}$ and parameters $k$,$\phi_{in}$,$\phi_{out}$,$\epsilon$, in which $\phi_{in},\phi_{out}$ satisfy $\phi_{out}=O(\frac{\epsilon^4}{\log{n}}\phi_{in}^2)$, with probability at least $2/3$,
    \begin{itemize}[topsep=3pt, partopsep=3pt, itemsep=3pt]
        \item[--] \textbf{accept} if there exists a $(k,\phi_{in},\phi_{out})$-cluster on $G$,
        \item[--] \textbf{reject} if $G$ is $\epsilon$-far from having a $(k,\phi_{in},\phi_{out})$-cluster,
    \end{itemize}
    where $\epsilon$-far means that $G$ cannot be accepted by modifying (inserting or deleting) no more than $\epsilon d_{max}n$ edges.
\end{definition}

Authors in \cite{czumaj2015} gave a detailed explanation on why they chose a logarithmic gap between $\phi_{in}^2$ and $\phi_{out}$. This paper maintains the gap of \cite{czumaj2015} in testing higher-order cluster structure and show that the $1$-dimension cluster is compatible with $(k,\phi_{in},\phi_{out})$-cluster in Theorem~\ref{the:1d}.

\subsection{Testing higher-order clusterability on bounded-degree graphs}
A few concepts on simplicial complex would be introduced before showing the definition~\ref{def:par} of higher-order cluster and the problem~\ref{def:test} of testing higher-order clusterability. These concepts could help us understand the graph in a high dimensional view. A \textit{d-simplex}~\cite{spanier1981} is the simplest geometric figure in $d$ dimension, e.g., point (0-simplex), line segment (1-simplex), triangle (2-simplex) and tetrahedron (3-simplex). A \textit{d-simplicial complex} $X$ is a collection of sets constructed by gluing together simplices with maximal dimension $d$. $X$ should satisfy a closure property that for any simplex $\sigma\in X$, all of its subsets $\tau\subset\sigma$ are also in $X$. $\sigma$ is denoted as a \textit{face} of $X$. Dimension of a face $dim(\sigma)$ equals to the number of vertices in it minus 1, i.e., $dim(\sigma)=|\sigma|-1.$ Empty set satisfies $\emptyset\in X$ with dimension $-1$ to keep closure. Other definitions are shown as follows: 
\begin{itemize}[leftmargin=*, topsep=5pt, partopsep=10pt, itemsep=5pt, label=--]
    \item \textit{$i$-faces $X_i$} is a set of all faces with dimension $i$.
    \item \textit{$i$-cochain $C(i)$} is a subset of $X_i$. Space of $i$-cochain is $S^i(X)$.
    \item \textit{Degree of face} $deg_d(\sigma)$ is the number of $d$-dimension faces that contain $\sigma$.
    \item \textit{Volume of $i$-cochain} $Vol_d(C(i))=\sum\limits_{\tau\in C(i)}deg_d(\tau)$.
    \item \textit{Norm of $i$-cochain} $\|C(i)\|_d=\frac{Vol_d(C(i))}{Vol_d(X_i)}$.
    \item \textit{Adjacent $i$-dimension faces} $a\sim b$ means there exists a face $\tau\in X_{i+1}$ that $a,b\subset \tau$.
    \item \textit{Induced $(i+1)$-subcomplex} $C(i)[X_{i+1}]=(C(i), \{\sigma\in X_{i+1}|\exists\tau\in\sigma: \tau\in C(i)\})$.
\end{itemize}
Kaufman and Mass~\cite{kaufman2017} proposed a high dimensional expander as follows,
\begin{definition} (\textbf{Colorful Expander~\cite{kaufman2017}}).
Let $X$ be a d-dimension simplicial complex. $X$ is an $\epsilon$-colorful expander, $\epsilon>0$, if for any $i$-cochain $C(i)\in S^i(X), 0\leq i<d, 0<\|C(i)\|_d\leq 1/2$,
    $$\frac{\|\mathbb{F}(C(i),X_i\backslash C(i))\|_d}{\|C(i)\|_d}\geq\epsilon,$$
    where $\mathbb{F}(C(i), X_i\backslash C(i))$ is the \textbf{expander face} (similar to cut on graphs) that is defined as
    $$\mathbb{F}(C(i), X_i\backslash C(i))=\{\sigma\in X_{i+1}|\exists\tau,\tau'\subset\sigma:\tau\in C(i),\tau'\in X_i\backslash C(i)\}.$$
\end{definition}
Similar to the internal and external conductance on undirected graphs, a normalized version of conductance is extended to simplicial complex.
\begin{definition} (\textbf{Normalized External Conductance}). 
    Let $X$ be a $d$-dimensi-on simplicial complex, $d\geq 1, 0\leq i<d$, $C(i)$ and $S(i)$ are both i-cochains, $\emptyset\neq S(i)\subset C(i)\subseteq X_i$, the normalized external conductance of $S(i)$ on $C(i)$ equals to $$\Psi_{d,C(i)}\{S(i)\}=\frac{\|\mathbb{F}(S(i),C(i)\backslash S(i))\|_d}{\min\{\|S(i)\|_d,\|C(i)\backslash S(i)\|_d\}}.$$
\end{definition}

\begin{definition} (\textbf{Normalized Internal Conductance}). 
    Let $X$ be a d-dimensi-on simplicial complex, $d\leq 1, 0\leq i<d$, $C(i)$ is an i-cochain, $\emptyset\neq C(i)\subseteq X_i$. The normalized internal conductance of $C(i)$ is $$\Psi_d(C(i)[X_{i+1}])=\min\limits_{\substack{\emptyset\neq S(i)\subset C(i)\\ Vol_d(S(i))\leq Vol_d(C(i))/2}}\frac{\|\mathbb{F}(S(i),C(i)\backslash S(i))\|_d}{\|S(i)\|_d}.$$
\end{definition}
The final step is to establish a unique mapping from simple undirected graph to d-dimension simplicial complex, which is easy to implement since the process can be seen as dimension raising.
\begin{lemma}\label{lem:raise}
    Given an undirected graph $G(V,E)$ and integer $d>1$, there exists a unique d-dimension simplicial complex $X^d(G)=\{X_0(G),X_1(G), X_2(G),\dots, X_{d}\\(G)\}$ that satisfies $X_0(G)=V, X_1(G)=E$, for each $i$, $1<i\leq d$, 
    $$X_i(G)=\{\bigcup(s_1,s_2,\dots,s_{i+1})|s_j,s_k\in X_{i-1}(G):s_j\sim s_k,\forall 1\leq j<k\leq i+1\}.$$
\end{lemma}
More generally, $X^d(G)$ is constructed by gluing together all $i$-cliques (triangles when $i=3$) to be its $(i-1)$-faces. The formal definition of high-dimension cluster that mentioned in the abstract is as follows,
\begin{definition}  (\textbf{$d$-dimension ($k,\psi_{in},\psi_{out}$)-cluster})\label{def:par}
    Given an undirected gr-aph $G(V,E)$ with parameters $d,k,\psi_{in},\psi_{out}$, find an $h$-partition $\mathbb{P}$ of $V$, $\mathbb{P}=(P_1,P_2,\dots,P_h)$ with $1\leq h\leq k$, and for each $i,r$, $1\leq i \leq h$,$0\leq r<d$, $\Psi_d(X_{r}(G[P_i])[X_{r+1}(G)])\geq\psi_{in}$ and $\Psi_{d,X_{r}(G)}(X_{r}(G[P_i]))\leq\psi_{out}$.
\end{definition}
The problem of testing higher-order clusterability is defined as follows,
\begin{definition} (\textbf{Testing $d$-dimension ($k,\psi_{in},\psi_{out},\epsilon$)-clusterability}).\label{def:test}
    Given a neighbor oracle access to graph $G(V,E)$ with maximum degree at most $d_{max}$ and parameters $d$,$k$, $\psi_{in}$,$\psi_{out}$,$\epsilon$, in which $\psi_{in},\psi_{out}$ satisfies $\psi_{out}=O(\frac{\epsilon^4}{\log{n}}\psi_{in}^2)$, with probability at least $2/3$,
    \begin{itemize}[topsep=3pt, partopsep=3pt, itemsep=3pt]
        \item[--] \textbf{accept} if there exists a $d$-dimension $(k,\psi_{in},\psi_{out})$-cluster on $G$,
        \item[--] \textbf{reject} if $G$ is $\epsilon$-far from having a $d$-dimension $(k,\psi_{in},\psi_{out})$-cluster,
    \end{itemize}
    where $\epsilon$-far denotes $G$ cannot be accepted by modifying (insertion or deletion) no more than $\epsilon d_{max}n$ edges.
\end{definition}

\section{Analysis of Compatibility and Lower Bound}\label{sec:comp}
\subsection{Compatibility with framework of testing clusterability}
The relationship between $1$-dimension $(k,\psi_{in},\psi_{out})$-partiton and $(k,\phi_{in},\phi_{out})$-partiton~\cite{gharan2014} is shown as follows,
\begin{theorem}\label{the:1d}
    $1$-dimension ($k,\psi_{in},\psi_{out}$)-cluster is equivalent to ($k,\frac{\psi_{in}}{2},\frac{\psi_{out}}{2}$)-cluster on undirected graph.
\end{theorem}

\begin{proof}
    For $1$-dimension $(k,\psi_{in},\psi_{out})$-partiton, $1\leq i\leq h$, $X_0=V$ and $X_1=E$, so $X_0(P_i[G])[X_1(G)]=P_i[G]$, $X_0(P_i[G])=P_i$ and $X_0(G)=V$. Therefore,
    $$\psi_{out}\geq\Psi_{d,X_{r}(G)}(X_{r}(P_i[G]))=\frac{|E(P_i,V\backslash P_i)|/|E|}{Vol(P_i)/Vol(V)}=2\Phi_V(P_i).$$
    Similarly, $\psi_{in}\leq\Psi_d(X_{0}(P_i[G])[X_{1}(G)])=2\Phi(P_i[G])$. The proof is finished by combining these two inequalities. 
\end{proof}

According to Theorem~\ref{the:1d}, algorithms for testing 1-dimension $(k,\psi_{in},\psi_{out},\epsilon)$-clusterability can also test $(k,\frac{\psi_{in}}{2},\frac{\psi_{out}}{2},\epsilon)$-clusterability in~\cite{czumaj2015}.
\subsection{Compatibility of high-dimension ($k,\psi_{in},\psi_{out}$)-cluster}
This section mainly deals with undirected graphs without outliers, which means all vertices or edges are contained in at least one triangle or $d$-clique. It is natural since if the graph has outliers, they can be eliminated without affecting quality of higher-order clustering. Following definition is necessary to prove compatibility,
\begin{definition} (\textbf{Induced $i$-graph}~\cite{kaufman2017}) Given a $d$-dimension simplicial complex $X$. For any $i$ with $0\leq i<d$, the $i$-graph $G_i(V_i, E_i)$ satisfies,
    \begin{itemize}[topsep=3pt, partopsep=3pt, itemsep=3pt]
        \item[1)] Every $i$-dimension face $\tau$ in $X_i$ is corresponding to a unique vertex $V(\tau)$.
        \item[2)] There is an edge between the corresponding vertex for any two adjacent $i$-dimension faces $\tau, \tau'$, i.e., $E_i=\{(V(\tau),V(\tau'))|\tau\sim\tau'\}$ 
    \end{itemize} 
\end{definition}
Generally speaking, induced $i$-graph is a dimensional reduction that maps the complex constructed by two $i$-faces $(X_i,X_{i+1})$ to an undirected graph. Corresponding to the graph without outliers, \textit{pure simplicial complex} $X$ is adopted that for any face $\tau\in X$ with $dim(\tau)<dim(X)$, there exists a face $\sigma\in X$, $dim(\sigma)=dim(X)$, such that $\tau\subset\sigma$. Then the following lemma holds,

\begin{lemma}\label{lem:norm}
    Let $X$ be a pure $d$-dimension simplicial complex. Given $t$ that satisfies $1\leq t<d$, for any $i$-cochain $C(i)$ that satisfies $0\leq i<t$, the external conductance is equal to $\frac{\|\mathbb{F}(C(i),X_i\backslash C(i))\|_d}{\|C(i)\|_d}.$ (see Appendix~\ref{lem2})
\end{lemma}

\begin{lemma}\label{lem:color}
    Let $X$ be a pure $d$-dimension $\epsilon$-colorful expander, then for any $t$ that $1\leq t<d$, $X$ must be a $t$-dimension $\epsilon$-colorful expander. (see Appendix~\ref{lem3})
\end{lemma}
Through the above two lemmas, it can be proved that if there exists a good cluster in high dimension, it is exactly a good cluster in lower dimension. 

\begin{theorem}\label{the:one}
    Given an undirected graph $G(V,E)$ without outliers, if $h$-partition $\mathbb{P}$ is a $d$-dimension ($k,\psi_{in},\psi_{out}$)-cluster with $d\geq 2, 1\leq h\leq k$, then $\mathbb{P}$ must be a $t$-dimension ($k,\psi_{in},\psi_{out}$)-cluster for all $t$ that satisfies $1\leq t\leq d-1$. (see Appendix~\ref{the2})
\end{theorem}

\subsection{Lower bound of testing higher-order clusterability}
\begin{theorem}
    With neighbor query oracle access, testing $d$-dimension ($k,\psi_{in},\\\psi_{out},\epsilon$)-clusterability on bounded-degree graph with neighbor query oracle has a lower bound $\Omega(\sqrt{n})$.
\end{theorem}
\begin{proof}
    Consider the special case when $k=1$. The origin testing problem would reduce to testing $d$-dimension $\psi$-colorful expansion, while any $\psi_{out}>0$ could be satisfied. Consider an undirected graph $G$ without outliers, which means a pure $d$-dimension simplicial complex $X$ can be constructed on it. According to Lemma~\ref{lem:color} and Theorem~\ref{the:1d}, if $X$ is a pure $d$-dimension $\psi_{in}$-colorful expander, $(X_0,X_1)$ should be a $1$-dimension $\psi_{in}$-colorful expander, which means $G$ is a normal $\frac{\psi_{in}}{2}$-expander. Goldreich and Ron~\cite{goldreich1997} proved that testing expansion on bounded degree graphs with has an $\Omega(\sqrt{n})$ lower bound. Suppose that there exists an algorithm that can test $d$-dimension $(1,\psi_{in},\psi_{out})$-clusterability in $o(\sqrt{n})$ queries, it can also answer the expansion test in $o(\sqrt{n})$ queries, which is a contradiction. To conclude, query lower bound of testing higher-order clusterability is $\Omega(\sqrt{n})$.
\end{proof}

In the next section, we would give an approach on triangle-based clusterability that could reach this lower bound.

\begin{algorithm}\label{alg:trw}
    \caption{2-dimension Random Walk (2DRW)}
    \KwIn{Initial vertex $v_0$ or edge $e_0$, length $l$.}
    \KwOut{$(v_0,v_1,\dots,v_l)$ if input is vertex; $(e_0,e_1,\dots,e_l)$ if input is edge.}
    \For{Step $t\in [0,l-1]$}{
        \If{Move from vertex $v_t$}{
            \For{Each neighbor $u_t$ of $v_t$}{
                Search all neighbors of $u_t$ and $v_t$\;
                Set the number of common neighbors $c(u_t)$ to $u_t$\;
            }
            Choose $u'_t$ with probability $\frac{c(u'_{t})}{\sum_{u_t\sim v_t}c(u_t)}$ as $v_{t+1}$ and move to it.\;
        }
        \ElseIf{Move from edge $e_t=(x_t,y_t)$}{
            Search all neighbors of $x_t$ and $y_t$\;
            \For{Each common neighbor $z_t$}{
                Put edges $(x_t,z_t)$ and $(y_t,z_t)$ into candidate set $C(e_t)$\;
            }
            Choose $e_{t+1}$ from set $C(e_t)$ uniformly at random and move to it\;
        }
    } 
\end{algorithm}

\section{Algorithm of Testing Triangle-based Clusterability}\label{sec:Algo}
\subsection{Design of Triangle-based k-cluster tester}
This section would give an example how to recognize triangle-based clusterability in sublinear-time with neighbor query oracle access. High-order random walk, which is used to catch information of network motifs, would be invoked in our algorithm. Related definition is shown as follows,
\begin{definition} (\textbf{High-order Random Walk ~\cite{kaufman2017}})
    Given a simplicial complex $X$ with $d$-dimension higher than $i$, the $i$-dimension high-order random walk $W_i$ starts from an initial $i$-dimension face $\tau_0\in X_i$. Then let $\tau_t$ be the position $W_i$ stays after $t$ steps. Choose $\tau_{t+1}$ as follows
    \begin{itemize}[itemsep=3pt, topsep=3pt, partopsep=3pt]
        \item[1)] Choose an $(i+1)$-dimension $\sigma_t\supset\tau_t$ with probability proportional to its degree $deg_d(\sigma_t)$.
        \item[2)] Uniformly choose an $i$-dimension face $\tau_{t+1}\subset\sigma_{t}, \tau_{t+1}\neq\tau_t$ at random and move to it.
    \end{itemize}
    $W_i$ stops at $\tau_t$ if no $\sigma_t$ or $\tau_{t+1}$ exists.
\end{definition}
The exact probability for moving from $\tau$ to $\tau'$, where $\tau\sim\tau'$, is as follows:
$$Pr[\tau_{t+1}=\tau'|\tau_t=\tau]=\frac{deg_d(\tau\cup\tau')}{\sum_{\tau''\sim\tau}deg_d(\tau\cup\tau'')}$$
Generally speaking, high-order random walk is an up-down Markov chain that moves on the induced $i$-subcomplex ($X_i$,$X_{i+1}$). Also, this random walk is equivalent to simple random walk on induced $i$-graph~\cite{kaufman2017} with probability distribution $\pi_0,\pi_1,\dots \in \mathbb{R}^{|X_i|}$ and $\pi_{t+1}=\pi_t\cdot\tilde{A_i}$, where $\tilde{A_i}$ is the normalized adjacency matrix of the $i$-graph. Thus, the distribution becomes stable when it equals to one of the eigenvectors of $\tilde{A_i}$. A complex with high expansion should satisfy that any high-order random walk converges rapidly to the uniform distribution.

However, neighbor query oracle cannot directly catch $i$-dimension face, so it is necessary to simulate this process by using more queries for each moving step. A 2-dimension random walk sampler in~\ref{alg:trw} is implemented for testing triangle-based clusterability. Given a vertex or edge as input, this sampler could perform the same up-down walk as that on  the induced $0$-graph and $1$-graph. Transition probability is proportional to the degree of the pass edge or triangle. Note that if no common neighbor exists, which means it is an outlier, the sampler would stop here as an endpoint.

\begin{algorithm}\label{alg:main}
    \caption{Triangle-based k-cluster tester}
    \KwIn{Query oracle of undirected bounded-$d_{max}$ graph $G(V,E)$, maximum cluster $k$, error $\epsilon$}
    \KwOut{Decision \textbf{Accept} or \textbf{Reject}}
    Sample a set $S_0$ of $s$ vertices independently and uniformly at random with query oracle\;
    For each $v\in S$, perform $m$ times of lazy $2DRW(v,l)$ and calculate the distribution $\pi_u^{l}$ of the endpoints\;
    \If{$k$-$cluster$-$test(\pi^{l},|V|,k,s,m,\theta,\delta,\epsilon)$ rejects}{
        Abort and return $Reject$\;
    } 
    \Else{
        Sample a set $S_1$ of $s$ edges independently and uniformly at random with edge sampler $S(G,\eta)$\;
        For each $e\in S$, perform $m$ times of lazy $2DRW(e,l)$ and calculate the distribution $\pi_e^{l}$ of the endpoints\;
        return $k$-$cluster$-$test(\pi^{l},|E|,k,2s,m,\theta,\delta,\epsilon)$\;
    }
\end{algorithm}

Here we briefly introduce our algorithm in~\ref{alg:main}. Similar to the approach in~\cite{czumaj2015} and~\cite{chiplunkar2018}, the algorithm embeds samples of vertices or edges into points on Euclidean spaces and cluster them based on the estimates of Euclidean distances. There are two main differences between our method and former ones. First, it is a two-step approach with $k$-$cluster$-$tester$~\ref{alg:kct} that tests whether the distribution vectors can be embedded into no more than $k$ clusters on Euclidean space. Second, simulated high-order random walks, which is promised to converge rapidly in high-dimension expander, is performed to estimate distribution of endpoints that reveals the similarities to each other. Note that \textit{lazy random walk} means with the probability $1/2$ for each step, the walk stay at the current vertex or edge. The edge sampler~\cite{eden2018} returns an edge that is n uniformly at random with bias $\eta$. Running time of the edge sampler is $O(\frac{n}{\sqrt{m}})$, which is $O(\sqrt{n})$ on bounded-degree graphs. Our algorithm would use the same configuration as that in~\cite{chiplunkar2018} that $\eta=\frac{1}{2}$ and number of edge samples would be doubled.

\begin{algorithm}\label{alg:kct}
    \caption{k-cluster-test}
    \KwIn{Distribution of endpoints $\pi^l$, maximum set size $n$, maximum cluster $k$, sample size $s$, number of each distribution $m$, parameters $\theta,\delta,\epsilon$}  
    \KwOut{Decision \textbf{Accept} or \textbf{Reject}}
    Similarity Graph $H=(\emptyset,\emptyset)$\;
    For each $u\in S$, if $l_2^2$-$norm(\pi_v^l,\theta,m,s)$ rejects, abort and return \textbf{Reject}\;
    For each pair of $u,v\in S$, if $l_2$-$distribution(\pi_u^l,\pi_v^l,m,s,\delta,\epsilon)$ accepts, then add an edge $(u,v)$ to $H$\;
    If $H$ contains more than $k$ connected components, return \textbf{Accept}; Else, return \textbf{Reject}\;
\end{algorithm} 

\subsection{Correctness and Running Time Analysis}
Now we prove the correctness of our algorithm. Since high order random walks are different from simple random walk, it is essential to make sure that distributions of endpoints converge as the input of $k$-$cluster$-$test$.

\begin{lemma}\label{lem:mix}\textbf{(Mixing Rate~\cite{kaufman2017}.)} Given an undirected graph $G(V,E)$, $\tilde{A}$ is its normalized adjacency matrix, $1=\alpha_1\geq\alpha_2\geq\cdots\geq\alpha_n\geq -1$ the eigenvalues of $\tilde{A}$ and $\alpha=\max\{|\alpha_2|,|\alpha_{|V|}|\}$. Then for any initial probability distribution $\pi_0\in R^{|V|}$ and any $t\in\mathbb{N}$,
$$\|\pi^t-\pi\|_2\leq\sqrt{\frac{d_{max}}{d_{min}}}\alpha^t,$$
where $\pi^t$ is the probability distribution after $t$ steps of the random walk, $\pi$ is the stationary distribution, $d_{max}=\max\limits_{v\in V}\{deg(v)\}$ and $d_{min}=\min\limits_{v\in V}\{deg(v)\}$.
\end{lemma}

We prove that a lazy $2$-dimension random walk with 11 times number of the original steps is enough for $k$-$cluster$-$test$.
\begin{lemma}~\label{lem:colormix} Given an undirected graph $G(V,E)$, mixing rate of lazy $2$-dimension random walk is $\mu'$, for any initial probability distribution $\pi_0\in R^{|V|}$ and any $t\in\mathbb{N}$,
$$\|\pi^t-\pi\|_2\leq\sqrt{\frac{d_{max}}{d_{min}}}\mu'^{11t},$$
where $\pi^t$ is the probability distribution after $t$ steps of the random walk, $\pi$ is the stationary distribution, $d_{max}=\max\limits_{v\in V}\{deg(v)\}$ and $d_{min}=\min\limits_{v\in V}\{deg(v)\}$. (see Appendix~\ref{lem5})
\end{lemma}

Then Theorem~\ref{the:cor} can be deduced by the next lemma. In convenience, we set $\psi_{in}=\psi$ and $\psi_{out}=O(\epsilon^4\psi^2/\log(n))$. We say a given graph is \textit{$2$-dimension $(k,\psi)$-clusterable} if there exists a $2$-dimension $(k,\psi_{in},\psi_{out})$-cluster on it.
\begin{lemma}~\label{lem:kct} Given the same constants $c_{3.1},c_{4.2},c_{4.3}$ as those in~\cite{czumaj2015}, set $s=\frac{1536k\ln(18(k+1))}{\epsilon^2}, l=\frac{11max\{c_{4.2},c_{4.3}\}k^4\log(n)}{\psi^2}, m=384c_{3.1}s\sqrt{skn}\ln s, \theta=\frac{288sk}{n},\\ \delta=\frac{1}{24s^2}$. $k$-$cluster$-$test$ accepts $2$-dimension ($k,\psi$)-clusterable graph and rejects every graph $\epsilon$-far from being $2$-dimension ($k,\psi$)-clusterable with probability at least $\frac{5}{6}$. (see Appendix~\ref{lem6})
\end{lemma}
Then Theorem $\ref{the:cor}$ holds since our algorithm invoke $k$-$cluster$-$test$ twice.

\begin{theorem}\label{the:cor} (\textbf{Correctness}.) With proper setting of the parameters, algorithm~\ref{alg:main} can accept every $2$-dimension $(k,\psi_{in},\psi_{out})$-clusterable graph with probability at least $\frac{2}{3}$ and reject every graph $\epsilon$-far from being $2$-dimension $(k,\psi_{in},\psi_{out})$-clusterable with probability at least $\frac{2}{3}$.
\end{theorem}

\begin{theorem}\label{lem:run} (\textbf{Running Time}.)
    With proper setting of parameters, triangle-based k-cluster tester runs in time $O(\frac{\sqrt{n}k^7d_{max}^3(\ln{k})^{7/2}\ln{1/\epsilon}\ln{n}}{\psi_{in}^2\epsilon^5}).$
\end{theorem}

\begin{proof} First, the algorithm generates a sample set of $s$ vertices with query oracle and $s$ edges with edge sampler. Next, the algorithm performs $m$ random walks with step $l$ for all $s$ samples, while time for each step is $O(d_{max}^3)$. Then, the algorithm invoke $l_2^2$-$norm$ tester, which has running time $O(m)$, for each sample in $S$. Finally, the algorithm invoke $l_2$-$distribution$ tester with running time $O(m)$ for each pair of samples in $S$. To conclude, the total running time of the algorithm is $O(s\sqrt{n}+d_{max}^3sml+sm+s^2m)=O(\frac{\sqrt{n}k^7d_{max}^3(\ln{k})^{7/2}\ln{1/\epsilon}\ln{n}}{\psi_{in}^2\epsilon^5}).$
\end{proof}

\section{Summary and Future Work}\label{sec:Fut}
In this work, a problem of testing higher order clusterability is proposed based on the new definition of high-dimension cluster. Besides, an algorithm for testing triangle-based clusterability, which reaches the proved lower bound, is designed. In the future, we would seek the lower bound when the logarithmic gap between normalized internal and external conductance is eliminated. We also plan to develop new algorithms for testing clique-based clusterability with more powerful query and sample oracles.

%
%
\bibliographystyle{splncs04}
\bibliography{samplepaper}

\newpage

\begin{subappendices}
	\renewcommand{\thesection}{\Alph{section}}
	\section{Proofs}
	\subsection{Proof of Lemma \ref{lem:norm} (See Page~\pageref{lem:norm})}\label{lem2}
    \begin{proof} 
        Let $G_i(V_i,E_i)$ be the induced $i$-graph of $X$. For any $i$-face $\tau\in X_i$ and its corresponding $V(\tau)\in V_i$, it holds that
         $$\begin{aligned}
             deg_d(\tau) &=\frac{1}{d-i}\sum\limits_{\substack{\tau\subset\sigma\\ dim(\tau)=\dim(\sigma)-1}}deg_d(\tau)\\
             &=\frac{1}{(d-i)(i+1)}\sum\limits_{\substack{\tau\subset\sigma\\ dim(\tau)=\dim(\sigma)-1}}deg_d(\sigma\cup\sigma')\\
             &=\frac{1}{(d-i)(i+1)}deg(V(\tau))\\
             &=\frac{t-i}{d-i}\cdot\frac{1}{(t-i)(i+1)}deg(V(\tau))\\
             &=\frac{t-i}{d-i}deg_{t}(\tau).
         \end{aligned}$$
         Let $C(i)$ be an $i$-cochain that satisfies $\emptyset\neq C(i)\subseteq X_i$, the norm satisfies
         \begin{equation}\label{equ:1}
         \|C(i)\|_d=\frac{\sum_{\tau\in C(i)}deg_d(\tau)}{\sum_{\tau'\in X_i}deg_d(\tau')}=\frac{\sum_{\tau\in C(i)}deg_{t}(\tau)\cdot\frac{t-i}{d-i}}{\sum_{\tau'\in X_i}deg_{t}(\tau')\cdot\frac{t-i}{d-i}}=\|C(i)\|_{t}.
         \end{equation}
         Similarly, the $(i+1)$-cochain $\mathbb{F}(C(i),X_i\backslash C(i))$ satisfies 
         \begin{equation}\label{equ:2}
         \|\mathbb{F}(C(i),X_i\backslash C(i))\|_d=\|\mathbb{F}(C(i),X_i\backslash C(i))\|_t.
         \end{equation}
         Combining Equation~\ref{equ:1} and \ref{equ:2}, 
         $$\frac{\|\mathbb{F}(C(i),X_i\backslash C(i))\|_t}{\|C(i)\|_t}=\frac{\|\mathbb{F}(C(i),X_i\backslash C(i))\|_d}{\|C(i)\|_d},$$
         which finishes the proof.
     \end{proof}
	
	\subsection{Proof of Lemma \ref{lem:color} (See Page~\pageref{lem:color})}\label{lem3}
    \begin{proof}
        For any $r$-cochain $C(r)$ that satisfies $\emptyset\neq C(r)\subseteq X_r$ and $Vol_t(C(r))\leq Vol_t(X_r)/2$. By Equation~\ref{equ:2}, $\|C(r)\|_t=\|C(r)\|_d<1/2$, thus $Vol_d(C(r)) \\\leq Vol_d(X_r/C(r))$. According to the definition of $d$-dimension $\epsilon$-colorful expander, 
        \begin{equation}\label{equ:3}
            \begin{aligned}
                \Psi_t(C(r)[X_{r+1}])&=\frac{\|\mathbb{F}(C(r),X_r\backslash C(r))\|_t}{\|C(r)\|_t}\\
                &=\frac{\|\mathbb{F}(C(r),X_r\backslash C(r))\|_d}{\min\{\|C(r)\|_d,\|X_r/C(r)\|_d\}}\geq\epsilon,
            \end{aligned}
        \end{equation}
        which finishes the proof.
    \end{proof}

    \subsection{Proof of Theorem \ref{the:one} (See Page~\pageref{the:one})}\label{the2}
\begin{proof}
    Construct the $d$-dimensional simplicial complex $X^d(G)=\{X_0(G),X_1(G)\\, \dots, X_{d}(G)\}$ by using Lemma~\ref{lem:raise}. For arbitrary integer $i\in [1,h], r\in[1,t]$. With the help of Lemma~\ref{lem:norm}, the bound on t-dimension normalized external conductance is   
    \begin{equation}\label{equ:4}
        \begin{aligned}
            \Psi_{t,X_r(G)}(X_r(P[G])) &=\frac{\|\mathbb{F}(X_r(P[G]), X_r(G)/X_r(P[G]))\|_t}{\min\{\|X_r(P[G]\|_t,\|X_r(G)/X_r(P[G])\|_t\}}\\
            &=\frac{\|\mathbb{F}(X_r(P[G]), X_r(G)/X_r(P[G]))\|_d}{\min\{\|X_r(P[G]\|_d,\|X_r(G)/X_r(P[G])\|_d\}}\\
            &\leq\phi_{out}.
        \end{aligned}
    \end{equation}
    According to the definition of $d$-dimension $(k,\psi_{in},\psi_out)$-cluster, $X^d(P_i[G])$ is a $d$-dimension colorful expander. Then using Lemma~\ref{lem:color}, the bound on $t$-dimension normalized internal conductance for any $P_i$ with $i\in[1,h]$ is
    \begin{equation}\label{equ:5}
        \begin{aligned}
            \Psi_t(X_r(P_i[G])[X_{r+1}(P_i[G])])&\\
            &=\min\limits_{\substack{\emptyset\neq C(r)\subset X_r(P_i[G])\\ Vol_d(C(r))\leq Vol_d(X_r(P_i[G]))/2}}\Psi_t(C(r)[X_{r+1}(P_i[G])])\\
            &\geq\phi_{in}.
        \end{aligned}
    \end{equation}
    Using Inequality~\ref{equ:4} and \ref{equ:5}, $\mathbb{P}$ is a $t$-dimension $(k,\psi_{in},\psi_{out})$-partiton, which finishes the proof.
\end{proof}

\subsection{Proof of Lemma \ref{lem:colormix} (See Page~\pageref{lem:colormix})}\label{lem5}
The following lemma gives a mixing rate for 2-dimension random walks on $\epsilon$-colorful expander.
\begin{proof}
    \begin{lemma}~\label{lem:rate} \textbf{(Mixing Rate on colorful expansion~\cite{kaufman2017}.)} Let $X$ be a $d$-dimension $\epsilon$-colorful expander, $d>1$. Then all high order random walks on $X$ are $\mu$-rapidly mixing for 
        $$\mu=1-\frac{\epsilon^2}{2(d+1)^2}$$
        where $\mu\leq\max\{|\alpha_2|,|\alpha_{|V|}|\}$ with the spectrum of the normalized adjacency matrix $\tilde{A}$ on each induced i-graphs. 
        \end{lemma}
        
        To keep the absolute value of the second eigenvalue larger than that of the last eigenvalue, lazy random walk is used with mixing rate $\mu'=\frac{1+\mu}{2}$. Lemma~\ref{lem:rate} reveals that $\mu'= 1-\frac{\epsilon^2}{36}$ for lazy 2-dimensional random walks (2DRW), while mixing rate of the lazy random walk on $\epsilon$-expander is $1-\frac{\epsilon^2}{4}$. Therefore, 2DRW needs more steps to converge.
        \begin{lemma}\label{lem:ine11} For any $\epsilon$ with $0<\epsilon\leq 1$, the following inequality holds:
            $$1-\frac{\epsilon^2}{4}>(1-\frac{\epsilon^2}{36})^{11}$$
        \end{lemma}
        \begin{proof}
            Let function $f(\epsilon)=1-\frac{\epsilon^2}{4}-(1-\frac{\epsilon^2}{36})^{11},$ so $f(0)=0$, $f(1)=\frac{3}{4}-(\frac{35}{36})^{11}>0$, the following equation holds,
            \begin{equation}
                \begin{aligned}
                    \frac{\partial{f(\epsilon)}}{\partial{\epsilon}} &=-\frac{\epsilon}{2}+11(1-\frac{\epsilon^2}{36})^{10}\cdot\frac{\epsilon}{18}
                    \\&=\frac{\epsilon}{18}\cdot(11(1-\frac{\epsilon^2}{36})^{10}-9)
                    \end{aligned}
            \end{equation}
            Therefore, $\epsilon=0.8457$ when $\partial{f(\epsilon)}/\partial{\epsilon}=0$, which means $f(\epsilon)$ monotonically increases when $\epsilon\in(0,0.8457]$ and decreases when $\epsilon\in(0.8457,1]$. To conclude, $f(\epsilon)>0$ when $\epsilon\in(0,1]$ and the inequality holds.
        \end{proof}
        According to Lemma~\ref{lem:ine11}, $\mu'^{11}$ is enough for convergence of endpoint distributions on $G$.
\end{proof}

\subsection{Proof of Lemma \ref{lem:kct} (See Page~\pageref{lem:kct})}\label{lem6}
\begin{proof}
    Here are the testers invoked in $k$-$cluster$-$test$: $l_2^2$-norm tester judges if a sufficiently long random walk from the vertex or edge can cover a large fraction of the graph; $l_2$-distribution tester distinguishes the Euclidean closeness of two distributions that is related to whether the given two vertices or edges are in the same cluster.

\begin{definition} \textbf{($l_2^2$-norm tester~\cite{czumaj2015}.)} Let $p$ be the probability distribution over a set of maximum size $n$. There exists an algorithm, called $l_2^2$-$norm(p_u^l,\theta,m)$, that takes $m$ samples of $p$ as input. It accepts the distribution if $\|p\|_2^2\leq\theta/4$ and rejects the distribution if $\|p\|_2^2>\theta$ with probability at least $1-\frac{16\sqrt{n}}{m}$. Running time of the tester is $O(m)$.
\end{definition}

\begin{definition} \textbf{($l_2$-distribution tester~\cite{czumaj2015}.)} Let $c_{3.1}$ be a constant with $c_{3.1}\geq 1$, $\delta,\xi >0$ and $p,q$ be two distributions over a set of size $n$ with $b>\max\{\|p\|_2^2,\\\|q\|_2^2\}$. Let $m>c_{3.1}\cdot\frac{\sqrt{b}}{\xi}\ln{\delta}$. There exists an algorithm, called $l_2^2$-$distribution$ tester, that takes as input $m$ samples from each distribution $p$,$q$, and accepts the distribution if $\|p-q\|_2^2\leq \xi$, and rejects the distributions if $\|p-q\|_2^2\leq 4\xi$. Running time of the tester is $O(m)$.
Note that $b$ and $\xi$ are implicit parameters that is used to set bounds on $m$. 
\end{definition}
Similar to the proof of Lemma 4.5$\sim$4.10 in~\cite{czumaj2015}, following lemmas hold with only values of parameters modified. On one hand, the completeness of $k$-$cluster$-$test$ is shown as follows,
\begin{lemma}\label{lem:accept} If the input graph $G$ is $2$-dimension ($k,\psi$)-clusterable, the algorithm $k$-$cluster$-$test$ accepts $G$ with probability at least $\frac{5}{6}$.
\end{lemma}
Now we give proof of Lemma~\ref{lem:accept}. Here we need the definition of \textit{good vertex (or edge)}.
\begin{definition} A vertex (or edge) $u$ is called good if the following three conditions are satisfied:
    \begin{itemize}[topsep=3pt, partopsep=3pt, itemsep=3pt]
        \item[1.] $\|p_{u}^l\|_2^2\leq\frac{72sk}{n}$.
        \item[2.] $|C(u)|\geq \frac{1}{36sk}n$, where $C(u)$ is the unique cluster that contains v.
        \item[3.] $v\in\tilde{C(u)}$, where $\tilde{C(u)}$ satisfies $|\tilde{C(u)}|\geq (1-\frac{1}{36s})|C(u)|$ and for any two vertices (or edges) $a$ and $b$ in $\tilde{C(u)}$, $\|p_a^l-p_b^l\|_2^2\leq \frac{1}{4n}$.
    \end{itemize}
\end{definition}
The following two lemmas can be deduced with the same process of Theorem 4.6 and 4.7 in~\cite{czumaj2015}. 
\begin{lemma}\label{lem:good} With probability at least $\frac{11}{12}$, all vertices in the sampled set $S$ are good.
\end{lemma}

\begin{lemma}\label{lem:goodtest} If all the sampled vertices $v\in S$ are good, $k$-$cluster$-$tester$ will accept $G$ with probability at least $\frac{11}{12}$.
\end{lemma}
Combining Lemma~\ref{lem:good} and~\ref{lem:goodtest}, proof of Lemma~\ref{lem:accept} is finished.

On the other hand, the soundness of $k$-$cluster$-$test$ is shown as follows,
\begin{lemma}\label{lem:reject} Let $\gamma_{d_{max},k}>0$ be some constant depending on $d_{max},k$. If the given graph $G$ is $\epsilon$-far from $(k,\psi*)$-clusterable with $\psi*\leq\frac{\gamma_{d_{max},k}\epsilon^2}{sl}$, then $k$-$cluster$-$test$ rejects $G$ with probability at least $\frac{5}{6}$.
\end{lemma}
Now we give proof of Lemma~\ref{lem:reject}. Here we need the definition of \textit{representative sample set}.
\begin{definition} Sample set $S$ on the disjoint sets $P_1,P_2,\dots P_{k+1}$ is said to be \textit{representative} if for all $i$ with $1\leq i\leq k+1$, $P_i\cap S\neq\emptyset$ and $S\subseteq\bigcup_{j=1}^{k+1}P_j$.
\end{definition}
The following two lemmas can be deduced with the same process of Theorem 4.9 and 4.10 in~\cite{czumaj2015}. 
\begin{lemma}\label{lem:repre} The probability that the sample set $S$ is representative is at least $\frac{11}{12}$.
\end{lemma}

\begin{lemma}\label{lem:bad} If $S$ is representative, then $k$-$cluster$-$test$ rejects $G$ with probability at least $\frac{23}{24}$. 
\end{lemma}
Combining Lemma~\ref{lem:repre} and~\ref{lem:bad}, proof of Lemma~\ref{lem:reject} is finished.

Now the proof of Lemma~\ref{lem:kct} follows directly from Lemma~\ref{lem:accept} and~\ref{lem:reject} by setting $\psi*\leq \psi_{out}$. 
\end{proof}
\end{subappendices}
\end{document}